\documentclass[11pt]{article}

\usepackage{amssymb,amsmath}

\usepackage{framed}

\usepackage{graphicx}

\usepackage{tikz}

\newtheorem{DE}{Definition}[section]

\newenvironment{proof}{\begin{trivlist} \item[] {\em Proof:}}{\hfill $
\Box$
                       \end{trivlist}}

\usepackage{latexsym} 
\usepackage{theorem} 
\newcommand{\qed}{\relax\ifmmode\hskip2em\Box\else\unskip\nobreak\hfill$\Box$\fi}

\newtheorem{theorem}[DE]{Theorem}
\newtheorem{lemma}[DE]{Lemma}

\newtheorem{corollary}[DE]{Corollary}

{\theoremstyle{break}\theorembodyfont{\rmfamily}}

\begin{document}

\title{Structure and algorithms for (cap, even hole)-free graphs}

\author{%
Kathie Cameron\thanks{Department of Mathematics, Wilfrid Laurier University,
Waterloo, ON, Canada, N2L 3C5. Research supported by the Natural Sciences and
Engineering Research Council of Canada (NSERC).}
\and%
Murilo V. G. da Silva\thanks{Department of Computer Science, Federal University of Technology-Paran\'{a}, Curitiba, Brazil.}
\and%
Shenwei Huang\thanks{School of Computer Science and Engineering,
University of New South Wales,  Sydney, NSW 2052, Australia.}
\and%
Kristina Vu\v{s}kovi\'c\thanks{School of Computing, University of
Leeds, Leeds LS2 9JT, UK; and Faculty of Computer Science (RAF), Union
University, Knez Mihajlova 6/VI, 11000 Belgrade, Serbia.  Partially
supported by EPSRC grant EP/N019660/1 and Serbian Ministry of
Education, Science and Technological Development projects 174033 and III44006.}}

\date{}

\maketitle

\begin{abstract}
A graph is even-hole-free if it has no induced even cycles of length 4 or more. A cap is a cycle of length at least
5 with exactly one chord and that chord creates a triangle with the cycle. In this paper, we consider (cap, even hole)-free graphs,
and more generally, (cap, 4-hole)-free odd-signable graphs. We give an explicit construction of these graphs. We prove that
every such graph $G$ has a vertex of degree at most
$\frac{3}{2}\omega (G) -1$, and hence $\chi(G)\leq \frac{3}{2}\omega (G)$, where $\omega(G)$ denotes the size of a largest clique in $G$
and $\chi(G)$ denotes the chromatic number of $G$.  We give an $O(nm)$ algorithm for $q$-coloring these graphs
for fixed $q$ and an $O(nm)$ algorithm for maximum weight stable set.  We also give a polynomial-time algorithm for minimum coloring.

Our algorithms are based on our results that triangle-free odd-signable graphs have treewidth at most 5 and thus
have clique-width at most 48, and that (cap, 4-hole)-free odd-signable graphs $G$ without clique cutsets have treewidth at most $6\omega(G)-1$
and clique-width at most 48.

\end{abstract}


\section{Introduction}

In this paper all graphs are finite and simple.
We say that a graph $G$ {\em contains} a graph $F$, if $F$ is
isomorphic to an induced subgraph of $G$.  A graph $G$ is
{\em $F$-free} if it does not contain $F$, and for a family of graphs ${\cal F}$,
$G$ is {\em ${\cal F}$-free} if $G$ is $F$-free for every $F\in {\cal F}$.
A {\em hole} is a chordless cycle of length at least four.
A hole is {\em even} (respectively, {\em odd}) if it contains an even (respectively, odd)
number of vertices.
A {\em cap} is a graph that consists of a hole $H$ and a vertex $x$ that has exactly
two neighbors in $H$, that are furthermore adjacent.
The graph $C_n$ is a hole of length $n$, and is also called an {\em $n$-hole}.
In this paper we study the class of (cap, even hole)-free graphs, and more generally
the class of (cap, 4-hole)-free odd-signable graphs, which we define later.

Let $G$ be a graph.
A set $S\subseteq V(G)$ is a {\em clique} of $G$ if all pairs of vertices of $S$ are adjacent.
The size of a largest clique in a graph $G$ is denoted by $\omega (G)$, and is
sometimes called the {\em clique number} of $G$.
We say that $G$ is a {\em complete graph} if $V(G)$ is a clique.
We denote by $K_n$ the complete graph on $n$ vertices.
The graph $K_3$ is also called a {\em triangle}.
A set $S \subseteq V(G)$ is a  {\em stable set} of $G$ if no two vertices
of $S$ are adjacent.
The size of a largest stable set of $G$ is denoted by $\alpha (G)$.
A {\em $q$-coloring} of $G$ is a function $c:V(G)\longrightarrow \{ 1, \ldots ,q\}$, such that
$c(u)\neq c(v)$ for every edge $uv$ of $G$.
The {\em chromatic number} of a graph $G$, denoted by
$\chi (G)$, is the minimum number $q$ for which there exists a $q$-coloring of $G$.



The class of (cap, odd hole)-free graphs has been studied extensively in
literature. This is precisely the class of Meyniel graphs, where a
graph $G$ is {\em Meyniel} if every odd length cycle of $G$,
that is not a triangle, has at least two chords.
These graphs were proven to be perfect by Meyniel \cite{meyniel}, and
Markosyan and Karapetyan \cite{mk}.
Burlet and Fonlupt \cite{bf} obtained the first
polynomial-time recognition algorithm
for Meyniel graphs, by decomposing these  graphs with amalgams
(that they introduced in the same paper).
Subsequently, Roussel and Rusu \cite{rr} obtained a faster algorithm for
recognizing Meyniel graphs (of complexity
$O(m^2)$), that is not decomposition-based.
Hertz \cite{hertz} gave an $O(nm)$ algorithm
for coloring and obtaining a largest clique of a Meyniel graph.
This algorithm is based on contractions
of even pairs. It is an improvement
on the $O(n^8)$ algorithm of Ho\`ang \cite{hoangM}.
Roussel and Rusu \cite{rr2} gave an $O(n^2)$ algorithm that colors
a Meyniel graph without using even pairs.
This algorithm ``simulates'' even pair contractions and it is based on
lexicographic breadth-first search and greedy sequential coloring.

Algorithms have also been given which find a minimum coloring of a Meyniel
graph, but do not require that the input graph be known to be Meyniel. A
\emph{Meyniel obstruction} is an induced subgraph which is an odd
cycle with at most one chord. A {\em strong stable set} in a
graph $G$ is a stable set which intersects every
(inclusion-wise) maximal clique of $G$. Cameron and Edmonds \cite{ce} gave an
$O(n^2)$ algorithm which for any graph, finds either a strong stable set or
a Meyniel obstruction. This algorithm can be applied at most $n$ times to find, in any graph,
either a clique and coloring of the same size or a Meyniel obstruction. Cameron,
L\'ev\^eque and Maffray \cite{clm} showed that a variant of the Roussel-Rusu coloring
algorithm for Meyniel graphs \cite{rr2} can be enhanced to find, for any input graph,
either a clique and coloring of the same size or a Meyniel obstruction. The worst-case
complexity of the algorithm is still $O(n^2)$.

In \cite{cckv-cap}, Conforti, Cornu\'ejols, Kapoor and Vu\v{s}kovi\'c
generalize Burlet and Fonlupt's decomposition theorem for Meyniel graphs
\cite{bf} to the decomposition by amalgams of all cap-free graphs.
One of the consequences of this theorem are polynomial-time recognition
algorithms for cap-free odd-signable graphs and  (cap, even hole)-free graphs.
Since triangle-free graphs are cap-free, it follows that the problems
of coloring and of finding the size of a largest stable set are both
NP-hard for cap-free graphs.
In \cite{cgp}, Conforti, Gerards and Pashkovich show how to obtain a polynomial-time
algorithm for solving the maximum weight stable set problem
on any class of graphs that is decomposable by amalgams into basic graphs
for which one can solve the
maximum weight stable set problem in polynomial time.
This leads to a first known non-polyhedral algorithm for maximum weight stable set problem
for Meyniel graphs.
Furthermore, using the decomposition theorems from \cite{cckv-cap} and \cite{cckv-tf},
they obtain a polynomial-time algorithm for solving the maximum weight
stable set problem  for
(cap, even hole)-free graphs (and more generally
cap-free odd-signable graphs).
For a survey on even-hole-free graphs and odd-signable graphs, see \cite{v}.

Aboulker, Charbit, Trotignon and Vu\v{s}kovi\'c \cite{actv} gave an $O(nm)$-time
algorithm whose input is a weighted graph $G$ and whose output is a maximum weighted clique of $G$
or a certificate proving that $G$ is not 4-hole-free odd-signable (the crux of this algorithm was actually developed
by da Silva and Vu\v{s}kovi\'c in \cite{daSV}).

In Section \ref{construction}, we give an explicit construction of (cap, 4-hole)-free odd-signable graphs,
based on \cite{cckv-cap} and \cite{cckv-tf}.
From this, in Section \ref{bound}, we derive
that every such graph $G$ has a vertex of degree at most
$\frac{3}{2}\omega (G) -1$, and hence $\chi(G)\leq \frac{3}{2}\omega (G)$. It follows that
$G$ can be colored with at most $\frac{3}{2}\omega (G)$ colors using the greedy coloring algorithm.

In Section \ref{cwtw}, we prove that triangle-free odd-signable graphs have treewidth at most 5 and thus
have clique-width at most 48 \cite{cr}. We also prove that (cap, 4-hole)-free odd-signable graphs $G$ without clique cutsets have
clique-width at most 48 and treewidth at most $6\omega(G)-1$.

In Section \ref{algorithms}, we give an $O(nm)$ algorithm for $q$-coloring
(cap, 4-hole)-free odd-signable graphs. We give a (first known)
polynomial-time algorithm for finding a minimum coloring of these graphs (chromatic number).
We also obtain an $O(nm)$  algorithm for the maximum weight
stable set problem for (cap, 4-hole)-free odd-signable graphs.
We observe that the algorithm in \cite{cgp} proceeds by first decomposing the graph by amalgams, a step that
takes $O(n^4m)$ time ($O(n^2m)$ to find an amalgam \cite{cc}, which is called on $O (n^2)$ times)
and creates $O (n^2)$ blocks that are then processed further.
Finally, we observe that all our algorithms are robust.

It is known that planar even-hole-free graphs have treewidth at most 49 \cite{sss}.
We observe that (cap, even hole)-free graphs are not necessarily planar.
Note that the graph in Figure \ref{figK5} is (triangle, even hole)-free and has
a $K_5$-minor.

The complexity of the stable set problem and the coloring problem remains open for even-hole-free graphs.




\begin{figure}
\center
\tikzstyle{every node}=[circle, draw, inner sep=0pt, minimum width=4pt]
\begin{tikzpicture}[scale=0.4]
\node [draw, circle](v1) at (0,5) {};
\node [draw, circle] (v2) at (4.5,1.5) {};
\node [draw, circle] (v3) at (4.5,-4) {};
\node [draw, circle] (v4) at (0,3.5) {};
\node [draw, circle] (v5) at (1,2.5) {};
\node [draw, circle] (v6) at (2,1.5) {};
\draw (v1) edge (v2);
\draw (v2) edge (v3);
\draw (v1) edge (v4);
\draw (v4) edge (v5);
\draw (v5) edge (v6);
\draw (v6) edge (v2);
\node [draw, circle] (v7) at (-1,2.5) {};
\node [draw, circle] (v8) at (-2,1.5) {};
\draw (v4) edge (v7);
\draw (v7) edge (v8);
\node [draw, circle] (v9) at (-4.5,1.5) {};
\draw (v1) edge (v9);
\draw (v9) edge (v8);
\node [draw, circle] (v10) at (-4.5,-4) {};
\draw (v9) edge (v10);
\draw (v10) edge (v3);
\node [draw, circle] (v13) at (2,-2.5) {};
\node [draw, circle] (v11) at (-2,-2.5) {};
\node [draw, circle] (v12) at (0,-2.5) {};
\draw (v10) edge (v11);
\draw (v11) edge (v12);
\draw (v13) edge (v12);
\draw (v13) edge (v3);
\node [draw, circle] (v14) at (2,0.5) {};
\node [draw, circle] (v15) at (2,-1) {};
\draw (v2) edge (v14);
\draw (v14) edge (v15);
\draw (v15) edge (v13);
\node [draw, circle] (v18) at (6,3) {};
\node [draw, circle] (v17) at (3.5,5) {};
\node [draw, circle] (v16) at (1,7) {};
\draw (v1) edge (v16);
\draw (v16) edge (v17);
\draw (v17) edge (v18);
\draw (v18) edge (v2);
\node [draw, circle] (v21) at (-1,7) {};
\node [draw, circle] (v20) at (-3.5,5) {};
\node [draw, circle] (v19) at (-6,3) {};
\draw (v9) edge (v19);
\draw (v19) edge (v20);
\draw (v20) edge (v21);
\draw (v21) edge (v1);
\node [draw, circle] (v23) at (6,-4) {};
\node [draw, circle] (v22) at (6,-0.5) {};
\draw (v18) edge (v22);
\draw (v22) edge (v23);
\draw (v23) edge (v3);
\node [draw, circle] (v25) at (-6,-4) {};
\node [draw, circle] (v24) at (-6,-0.5) {};
\draw (v24) edge (v25);
\draw (v25) edge (v10);
\draw (v24) edge (v19);
\node [draw, circle] (v26) at (4.5,-5.5) {};
\node [draw, circle] (v27) at (0,-5.5) {};
\node [draw, circle] (v28) at (-4.5,-5.5) {};
\draw (v3) edge (v26);
\draw (v26) edge (v27);
\draw (v27) edge (v28);
\draw (v28) edge (v10);
\node [draw, circle] (v29) at (-7.5,-5.5) {};
\node [draw, circle] (v30) at (-7.5,-0.5) {};
\draw (v28) edge (v29);
\draw (v29) edge (v30);
\draw (v30) edge (v9);
\end{tikzpicture}
\caption{\label{figK5} A (triangle even hole)-free graph that has a $K_5$-minor.}
\end{figure}
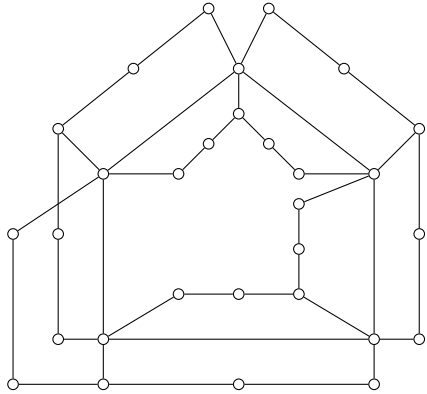



\section{Odd-signable graphs} \label{odd-signable}

We {\em sign} a graph by assigning $0,1$ weights to its edges.
A graph is {\em odd-signable} if there exists a signing that makes
 every chordless cycle odd weight. To characterize
odd-signable graphs in terms of excluded induced subgraphs, we now
introduce two types of {\em 3-path configurations} ($3PC$'s)
and even wheels.

Let $x$ and $y$ be two distinct vertices of $G$. A $3PC(x,y)$ is a graph
induced by three chordless $xy$-paths, such that any two of them induce
a hole. We say that a graph $G$ contains a $3PC(\cdot , \cdot )$ if
it contains a $3PC(x,y)$ for some $x,y \in V(G)$. The $3PC(\cdot , \cdot)$'s
are also known as {\em thetas}.

Let $x_1,x_2,x_3,y_1,y_2,y_3$ be six distinct vertices of $G$ such that
$\{ x_1,x_2,x_3\}$ and $\{ y_1,y_2,y_3\}$ induce triangles.
A $3PC(x_1x_2x_3,y_1y_2y_3)$ is a graph induced by three chordless paths
$P_1=x_1, \ldots ,y_1$, $P_2=x_2, \ldots ,y_2$ and $P_3=x_3, \ldots ,y_3$,
such that any two of them induce a hole. We say that a graph $G$ contains a
$3PC(\Delta , \Delta )$ if it contains a $3PC(x_1x_2x_3,y_1y_2y_3)$ for some
$x_1,x_2,x_3,y_1,y_2,y_3 \in V(G)$. The $3PC(\Delta , \Delta)$'s are also known
as {\em prisms}.

A {\em  wheel}, denoted by $(H,x)$, is  a graph induced by  a hole $H$
and a  vertex $x \not\in V(H)$  having at least three  neighbors in $H$,
say $x_1, \ldots  ,x_n$.
A
subpath of  $H$ connecting  $x_i$ and  $x_j$ is a  {\em sector}  if it
contains no intermediate  vertex $x_l$, $1 \leq l \leq  n$.
A wheel $(H,x)$ is {\em even} if it has an even number of sectors.

It  is easy to  see that  even wheels, thetas
and prisms cannot be  contained in even-hole-free
graphs. In fact they cannot be contained in odd-signable graphs.
The following characterization of odd-signable graphs states that
the converse  also holds, and it is an easy consequence of a theorem of
Truemper \cite{truemper}.

\begin{theorem} {\em\bf (\cite{cckv-cap})} \label{thm:forbid odd-signable}
A graph is odd-signable if and only if it does not contain
an even wheel, a theta or a prism.
\end{theorem}


\section{Construction of (cap, 4-hole)-free odd-signable graphs} \label{construction}

Let $G$ be a graph and $S\subseteq V(G)$.
The subgraph of $G$ induced by $S$ is denoted by $G[S]$, and $G\setminus S=G[V(G)\setminus S]$.
We say that $S$ is a {\em vertex cutset} of $G$ if $G\setminus S$ is disconnected.
A {\em clique cutset} of $G$ is a vertex cutset that is a clique of $G$. Note that an empty set is a clique,
and hence every disconnected graph has a clique cutset.

Let $A$ and $B$ be disjoint subsets of vertices of a graph $G$. We say that $A$ is {\em complete} to $B$
if every vertex of $A$ is adjacent to every vertex of $B$, and $A$ is {\em anticomplete} to $B$ if no vertex
of $A$ is adjacent to a vertex of $B$.

A connected graph $G$ has an {\em amalgam} $(V_1,V_2,A_1,A_2,K)$ if
the following hold:
\begin{itemize}
\item $V(G)=V_1 \cup V_2 \cup K$, where $V_1,V_2$ and $K$ are disjoint sets, and  $|V_1| \geq 2$, $|V_2| \geq 2$.
\item $K$ is a (possibly empty) clique of $G$.
\item For $i=1,2$, $\emptyset \neq A_i \subseteq V_i$.
\item $A_1$ is complete to $A_2$, and these are the only edges with one  end in $V_1$ and the other in $V_2$.
\item $K$ is complete to $A_1 \cup A_2$ (note that vertices of $K$ may have other neighbors in $V_1\cup V_2$).
\end{itemize}

A graph is {\em chordal} if it is hole-free.
A graph is {\em 2-connected} if it has at least 3 vertices and remains connected whenever fewer than 2 vertices are removed.
A {\em basic cap-free graph} $G$ is either a chordal graph or a
2-connected triangle-free graph together with at most one additional
vertex, that is adjacent to all other vertices of $G$.

\begin{theorem}\label{cap-decomp} {\em\bf (\cite{cckv-cap})}
A connected cap-free graph that is not basic has an amalgam.
\end{theorem}

A {\em module} (or {\em homogeneous set})
in a graph $G$ is a set $M \subseteq V(G)$, such that
$2\leq |M|\leq |V(G)|-1$
and every vertex of $V(G) \setminus M$ is either adjacent to all of $M$ or
none of $M$.
Note that a module with $2\leq |M|\leq |V(G)|-2$ (i.e. a {\em proper module})
is a special case of an amalgam.
A {\em clique module} is a module that induces a clique.

Let $G$ be a 4-hole-free graph and $(V_1,V_2,A_1,A_2,K)$
an amalgam of $G$. Without loss of generality we may assume that $A_1$ induces a clique, and hence
either $A_1 \cup K$ is a clique cutset or $A_1$ is a proper clique module.
So, Theorem \ref{cap-decomp} particularized to 4-hole-free graphs,
and the well known fact that every chordal graph that is
not a clique has a clique cutset, gives the
following decomposition.

\begin{theorem}\label{cfehf-decomp} 
If $G$ is (cap, 4-hole)-free graph,
then $G$ has a clique cutset or a proper clique module, or
$G$ is a complete graph or
a (triangle, 4-hole)-free graph together with at most one
additional vertex that is adjacent to all other vertices of $G$.
\end{theorem}

We now give a complete structural description of (cap, 4-hole)-free
graphs that do not have a clique cutset, by considering the proof of
Theorem \ref{cap-decomp} from \cite{cckv-cap} particularized to
4-hole-free graphs.

An {\em expanded hole} consists of nonempty disjoint sets of vertices $S_1, \ldots ,S_k$,
$k \geq 4$, not all singletons, such that for all $1 \leq i\leq k$,
the graphs $G[S_i]$ are connected. Furthermore, for $i\neq j$, $S_i$ is complete
to $S_j$ if  $j=i+1$ or $j=i-1$ (modulo $k$), and anticomplete otherwise.
We also say that $G[S]$ is an expanded hole.
Note that if an expanded hole is  4-hole-free, then $S_i$ is a clique
for every  $i=1, \ldots, k$.
The following statement can be extracted from the proofs of Lemma 5.1 and Theorem 7.1 in
\cite{cckv-cap}.

\begin{lemma}\label{cfehf-decomp2}
Let $G$ be a (cap, 4-hole)-free graph.
Suppose that $S=\cup_{i=1}^{k} S_i$ is  an inclusion-wise maximal expanded hole of $G$
such that $|S_2|\geq 2$. Let $U$ be the set of vertices of $G$ that are complete to $S$.
Then $G$ has an amalgam $(V_1,V_2,A_1,A_2,K)$ where $S_2=A_2$ and $K\subseteq U$.
In particular,
either $K \cup S_2$ is a clique cutset
or $S_2$ is a proper clique module.
\end{lemma}

Let $M$ be a proper clique module of a graph $G$. The block of decomposition
of $G$ with respect to this module is the graph $G'=G\setminus (M\setminus \{u\})$,
where $u$ is any vertex of $M$.
For a graph $G$ and a vertex $u$ of $G$, we denote $N_G(u)$ (or $N(u)$ when clear from context)
by the set of neighbors of $u$ in $G$. Also $N[u]=N(u)\cup \{ u\}$.
The degree $d_G(u)$ of $u$  is $|N_G(u)|$.

\begin{lemma}\label{cutsets}
Let $M$ be a proper clique module of a graph $G$, and let $G'$ be the block of
decomposition with respect to this module.
If $G$ does not have a clique cutset, then $G'$ does not have a
clique cutset.
\end{lemma}

\begin{proof}
Suppose $K$ is a clique cutset of $G'$. Let $u$ be the vertex of $M$ that
is in $G'$. If $u \not\in K$ then $N_{G'}(u)\setminus K$ is in the same
connected component of $G'\setminus K$ as $u$, and hence $K$ is a
clique cutset of $G$. If $u \in K$ then $M \cup K$ is a clique cutset
of $G$.
\end{proof}

\begin{theorem}\label{t1}
Let $G$ be a (cap, 4-hole)-free graph that contains a hole.
Let $F$ be a maximal vertex subset of $V(G)$ that induces a 2-connected
triangle-free graph,
$U$  the vertices of $V(G)\setminus F$ that are complete to $F$,
$D$ the vertices of $V(G)\setminus F$ that have at least
two neighbors in $F$ but are not complete to $F$, and
$S=V(G)\setminus (F \cup U \cup D)$.
Then the following hold:
\begin{itemize}
\item[(i)] $U$ is a clique.
\item[(ii)] $U$ is complete to $D \cup F$.
\item[(iii)] If $G$ does not have a clique cutset, then
for every $d \in D$, there is a vertex $u\in F$ and $D'\subseteq D$
that contains $d$ such that $D' \cup \{ u \}$ is a clique module of $G$.
In particular, for every $d'\in D'$, $N[d']=N[u]$.
\item[(iv)] If $G$ does not have a clique cutset, then $S=\emptyset$.
\item[(v)] If $G$ does not have a clique cutset, $F$ does not have a clique cutset.
\end{itemize}
\end{theorem}

\begin{proof}
Since $G$ is 4-hole-free and $F$ contains nonadjacent vertices, clearly
$U$ must be a clique, and hence (i) holds.

\vspace{2ex}

\noindent
{\bf Claim:} {\em For every $d \in D$, $G[F]$ contains a hole $H$
such that $G[V(H) \cup \{ d \}]$ is an expanded hole of $G$.}
\\
\noindent {\em Proof of the Claim:}
Let $d$ be any vertex of $D$, and assume that there is no hole H contained in $G[F]$ such that $G[V(H) \cup \{ d \}]$
is an expanded hole.
Note that $G[F \cup \{ d\}]$ contains a triangle $d,x,y$, for otherwise the maximality of
$F$ is contradicted.
Since $G[F]$ is 2-connected and triangle-free,
$x$ and $y$ are contained in a hole $H$ of $G[F]$.
It is easy to see that since $G$ is cap-free
and we are assuming that $G[V(H)\cup \{ d \}]$ is not
an expanded hole, it follows that $d$ is complete to $V(H)$.
Let $F'$ be a maximal subset of $F$ such that $G[F']$ contains $H$, is
2-connected and $d$ is complete to $F'$. Since $F\neq F'$ and both
$G[F]$ and $G[F']$ are 2-connected, some $z\in F \setminus F'$ belongs to
a hole $H'$ that contains an edge of $G[F']$. As before, it follows that
$d$ is universal for $H'$, and hence $F' \cup V(H')$ contradicts the
choice of $F'$.
This completes the proof of the Claim.

\vspace{2ex}

By the Claim, every vertex $d$ of $D$ has two nonadjacent neighbors in $F$,
and since $G$ is 4-hole-free, it follows that
every vertex of $U$ is adjacent to $d$.
Therefore, (ii) holds.

Now suppose that $G$ does not contain a clique cutset.
By the Claim  and Lemma \ref{cfehf-decomp2}, (iii) holds.
Let $D' \cup \{ u \}$ be a proper clique module from (iii). Then the block
of decomposition with respect to this module is the graph $G\setminus D'$.
So by performing a sequence of clique module decompositions, we get the
graph $G'=G \setminus D$. By Lemma \ref{cutsets}, $G'$ does not have a
clique cutset. Suppose that $S \neq \emptyset$.
Note that every vertex in $S$ has at most one neighbor in $F$.
Let $C$ be a connected component of $G[S]$. By the maximality of $F$,
there is at most one vertex in $F$, say $y$, that has a neighbor in $C$.
So $U \cup \{ y\}$ is a clique cutset of $G'$, a contradiction.
If no component of $G[S]$ is adjacent
to a vertex of $F$, then $U$ is a clique cutset of $G'$, a contradiction.
Therefore, (iv) holds.
Finally, suppose that $F$ has a clique cutset $K$. Then by (i) and (iv), $K \cup U$ is a clique
cutset of $G'$, a contradiction. Therefore, (v) holds.
\end{proof}

We say that the graph $G'$ is obtained from a graph $G$ by
{\em blowing up vertices of $G$ into cliques} if $G'$ consists of the disjoint
union of cliques $K_u$, for every $u \in V(G)$, and all edges between
cliques $K_u$ and $K_v$ if and only if $uv \in E(G)$. This is also referred to
as substituting clique $K_u$ for vertex $u$ (for all $u$).
The graph $G'$ is obtained from a graph $G$ by {\em adding a universal clique}
if $G'$ consists of $G$ together with (a possibly empty) clique $K$,
and all edges between vertices of $K$ and vertices of $G$.
Note that both of these operations preserve being (cap, 4-hole)-free,
i.e., $G$ is  (cap, 4-hole)-free if and only if $G'$ is (cap, 4-hole)-free.

\begin{theorem}\label{t2}
Let $G$ be a (cap, 4-hole)-free graph that contains a hole and has no clique
cutset.
Let $F$ be any maximal induced subgraph of $G$
with at least 3 vertices that is
triangle-free and has no clique cutset. Then $G$ is obtained from $F$ by first blowing up
vertices of $F$ into cliques, and then adding a
universal clique.
Furthermore, any graph obtained by this sequence of operations starting from
a (triangle, 4-hole)-free graph with at least 3 vertices and no clique cutset is (cap, 4-hole)-free
and has no clique cutset.
\end{theorem}

\begin{proof}
Let $F'$ be a maximal 2-connected triangle-free induced subgraph of $G$ that contains $F$.
By Theorem \ref{t1} (v), $F'$ does not have a clique cutset, and hence $F'=F$.
So the first statement follows from Theorem \ref{t1}.
The second statement follows from an easy observation that blowing up vertices into cliques
and adding a universal clique preserves being (cap, 4-hole)-free and having no clique cutset.
\end{proof}

Triangle-free odd-signable graphs were studied in \cite{cckv-tf} where the following
construction was obtained.
A chordless $xz$-path $P$ is an {\em ear} of a hole $H$ contained in a graph $G$
if $V(P)\setminus \{ x,z\} \subseteq V(G)\setminus V(H)$,
vertices $x,z\in V(H)$ have a common neighbor $y$ in $H$,
and $(V(H)\setminus \{ y\}) \cup V(P)$ induces a hole $H'$ in $G$.
The vertices $x$ and $z$ are the {\em attachments} of $P$ in $H$, and $H'$ is said to be
obtained by {\em augmenting} $H$ with $P$.
A graph $G$ is said to be obtained from a graph $G'$ by an {\em ear addition}
if the vertices of $G \setminus G'$ are the intermediate vertices of an ear
of some hole $H$ in $G'$, say an ear $P$ with attachments $x$ and $z$ in $H$,
and the graph $G$ contains no edge connecting a vertex of
$V(P)\setminus \{ x,z\}$ to a vertex of
$V(G')\setminus \{ x,y,z\}$, where $y\in V(H)$ is adjacent to $x$ and $z$.
An ear addition is {\em good} if
\begin{itemize}
\item $y$ has an odd number of neighbors in $P$,
\item $G'$ contains no wheel $(H_1,v)$, where $x,y,z\in V(H_1)$ and $v$ is
adjacent to $y$, and
\item $G'$ contains no wheel $(H_2,y)$, where $x,z$ are neighbors of
$y$ in $H_2$.
\end{itemize}

The complete bipartite graph $K_{4,4}$ with a perfect matching removed
is called a {\em cube}. Note that the cube contains 4-holes.

\begin{theorem} {\em\bf (Theorem 6.4 in \cite{cckv-tf})} \label{tfree}
Let $G$ be a  triangle-free graph with at least three vertices
that is not a cube and has no clique cutset. Then, $G$ is odd-signable
if and only if it can be obtained, starting from a hole, by a
sequence of good ear additions.
\end{theorem}


\section{Bound on the chromatic number} \label{bound}

We observe that the class of (triangle, 4-hole)-free graphs has
unbounded chromatic number.
In \cite{mgr} it is shown that (triangle, even hole)-free graphs
have a vertex of degree at most 2.
We now show that (triangle, 4-hole)-free odd-signable graphs
that contain at least one edge, have an edge both of whose ends are of degree
at most 2.
This will imply that
every (cap, 4-hole)-free odd-signable graph (and in particular every (cap, even hole)-free graph) $G$ has a vertex of
degree at most $\frac{3}{2}\omega (G)-1$, and hence that every graph in this
class has a proper coloring that uses at most $\frac{3}{2}\omega (G)$ colors.

\begin{theorem}\label{t22}
Every (cap, 4-hole)-free odd-signable graph $G$ has a vertex of
degree at most $\frac{3}{2}\omega (G)-1$.
\end{theorem}

\begin{proof}
Given a graph $G$, let us say that a vertex $v$ of $G$ is {\em nice} if
its degree is at most $\frac{3}{2}\omega (G)-1$.
We prove that {\em if $G$ is a (cap, 4-hole)-free odd-signable graph, then
either $G$ is complete or it has at least two nonadjacent nice vertices}.
Assume that this does not hold and let $G$ be a minimum counterexample.

Suppose that $G$ has a clique cutset $K$. Let $C_1, \ldots ,C_k$ be
the connected components of $G \setminus K$, and for $i=1,\ldots ,k$,
let $G_i=G[C_i \cup K]$. By the choice of $G$, for every $i$,
$G_i$ is either a complete graph or it has at least two nonadjacent nice vertices.
So there is a vertex $v_i \in C_i$ that is nice in $G_i$, and hence
in $G$ as well. But then $G$ has at least two nonadjacent nice vertices,
a contradiction. Therefore, $G$ does not have a clique cutset.

This also implies that $G$ cannot be chordal, since every chordal graph is either complete
 or has a clique cutset. So, $G$ contains a hole.
Let $F$ be a maximal induced subgraph of $G$ that is
triangle-free and has no clique cutset.
By Theorem \ref{t2}, $G$ is obtained from $F$ by blowing up vertices of $F$
into cliques and adding a universal clique $U$. Note that if a vertex
$u$ is nice in $G \setminus U$, then it is nice in $G$, and
hence by the choice of $G$, $U=\emptyset$.
For $u\in V(F)$, let $K_u$ be the clique that the vertex $u$ is blown up into.

\vspace{2ex}

\noindent
{\bf Claim:} {\em If $u_1,u,v,v_1$ is a path of $F$ such that $u$ and
$v$ are both of degree 2 in $F$, then $u$ or $v$ is nice in $G$}.

\vspace{2ex}

\noindent
{\em Proof of the Claim:} Since $|K_u|+ |K_v|\leq \omega (G)$, we may assume that
$|K_u|\leq \frac{1}{2}\omega (G)$. But then
$d_G(v)=|K_{u}|+|K_v|-1+|K_{v_1}|\leq \frac{3}{2}\omega (G)-1$.
This completes the proof of the Claim.

\vspace{2ex}

By Theorem \ref{tfree} we consider the last ear $P$ in the construction of
$F$. Say that $P$ is an ear of hole $H$ and its attachments in $H$ are
$x$ and $z$. Let $y$ be the common neighbor of $x$ and $z$ in $H$.
Since $P$ is a good ear, $y$ has an odd number of neighbors in $P$.
Let $y_1, \ldots ,y_k$ be the neighbors of $y$ in $P$ in the order when
traversing $P$ from $x$ to $z$ (so $y_1=x$ and $y_k=z$). Since $F$ is
(triangle, 4-hole)-free, both $y_1y_2$-subpath of $P$ and $y_{k-1}y_k$-subpath of
 $P$ contain an edge both of whose ends are of degree 2. It then
 follows by the Claim that $G$ has at least two
nonadjacent nice vertices, a contradiction.
\end{proof}

\begin{corollary}\label{cor:binding}
If $G$ is a (cap, 4-hole)-free odd-signable graph, then
$\chi (G) \leq \frac{3}{2}\omega (G)$.
\end{corollary}

\begin{proof}
It follows immediately from Theorem \ref{t22}.
\end{proof}

Theorem \ref{t22} immediately gives a $\frac{3}{2}$-approximation algorithm for coloring
(cap, 4-hole)-free odd-signable graphs.  The algorithm greedily colors a particular ordering
of vertices $v_1,v_2,\ldots, v_n$, where $v_i$ is a vertex of minimum degree in $G[v_1,\ldots,v_i]$.
It is clear that this ordering of vertices can be found in $O(n^2)$ time. 
 Therefore, Theorem \ref{t22} ensures that the greedy algorithm properly colors
the graph using at most $\frac{3}{2}\omega(G)$ colors in $O(n^2)$ time.


\section{Treewidth and clique-width} \label{cwtw}

A {\em triangulation} $T(G)$ of a graph $G$ is obtained from $G$ by adding edges until no
holes remain. Clearly, $T(G)$ is a chordal graph on the same vertex set as $G$ which contains all
edges of $G$. The {\em treewidth} of a graph $G$ is the minimum over all triangulations $T(G)$
of $G$ of the size of the largest clique in $T(G)$ minus 1. Treewidth $k$ is equivalent to having a tree decomposition
of width $k$, which is generally used in algorithms. Bodelander \cite{bod96} gave an algorithm which
for fixed $k$, recognizes graphs of treewidth at most $k$, and constructs a width $k$ tree decomposition;
the algorithm is linear in $n=|V(G)|$ but exponential in $k$.

The {\em clique-width} of a graph $G$ is the minimum
number of labels required to construct $G$ using the following four operations: creating a new vertex with
label $i$, joining each vertex with label $i$ to each vertex with label $j$, changing the label of every
vertex labelled $i$ to $j$, and taking the disjoint union of two labelled graphs. A sequence of
these operations which constructs the graph using at most $k$ labels is called a {\em k-expression} or
{\em clique-width k-expression}.

Treewidth and clique-width have similar algorithmic implications. Problems that can be expressed in monadic
second-order logic $MSO_2$ can be solved in linear time for any class of graphs with treewidth
at most $k$ \cite{courcelle}. Problems that can be expressed in the subset $MSO_1$ of $MSO_2$ which does
not allow quantification over edge-sets can be solved in linear time for any class of graphs of clique-width
at most $k$ \cite{cmr}. Kobler and Rotics \cite{kr} showed
that certain other problems including chromatic number can be solved in polynomial time for any class of graphs of
clique-width at most $k$ \cite{kr}.  All these algorithms are exponential in $k$, and those using clique-width
require a $k$-expression as part of the input. This latter requirement was removed by Oum and Seymour \cite{os}
and then more efficiently by Oum \cite{oum}, who gave an $O(n^3)$ algorithm  which, for any input graph and fixed integer $k$,
finds a clique-width $8^k$-expression or states that the clique-width is greater than $k$.

Corneil and Rotics \cite{cr} improved a result of Courcelle and Olariu \cite{co} to show that
the clique-width of a graph $G$  is at most $3\times 2^{tw(G)-1}$. Further, they construct, in polynomial time, a $k$-expression
of the stated size. Espelage, Gurski, and Wanke \cite{egw} gave an algorithm that takes as input a tree decomposition
with width $k$, and gives a clique-width $2^{O(k)}$-expression \cite{egw} in linear time.

\begin{theorem} \label{tftw5}
Every triangle-free odd-signable graph has treewidth at most 5.
\end{theorem}

\begin{proof}
Let $G^*$ be a triangle-free odd-signable graph.
 Decompose $G^*$ by clique cutsets into atoms that have no clique cutsets.
 We will show that each atom is contained in a chordal graph that has maximum clique size at most 6.
 Gluing these chordal graphs together along the clique cutsets used to decompose $G^*$ gives a
 chordal graph containing $G^*$ with clique size at most 6, and this proves the theorem.

Clearly every cube is contained in a chordal graph with maximum clique of size 5.
So we may assume that an atom $G$ has at least 3 vertices and is not a cube.
Then by Theorem \ref{tfree}, $G$ can be obtained from a hole $H$ by a sequence of good ear additions.
Let $P_1,\dots, P_q$ be the sequence of ears in the construction, with $P_q$ being the last ear added.
For each $i$, let $H_i$ be the hole $P_i$ is attached to, let $x_i$ and $z_i$ be the attachments of $P_i$ in $H_i$,
and let $y_i$ be the common neighbor of $x_i$ and $z_i$ in $H_i$.

We now obtain a triangulation $T$ of $G$ whose largest clique has size at most 6.
We construct $T$ as follows.
For each ear $P_i$, make $x_i, y_i$ and $z_i$ complete to $P_i \setminus \{x_i,z_i\}$, and add the edge $x_iz_i$. Call the edges $x_iz_i$ type 1 edges.
Choose any edge $uv$ of $H$, and join $u$ and $v$ to all the vertices of $H \setminus \{u,v\}$; call these type 2 edges.
For each $i$, let $S_i = \{x_i, y_i, z_i\}$,  $G_i = G[V(H) \cup V(P_1) \cup \dots \cup V(P_i)]$ and $T_i = T[V(H) \cup V(P_1) \cup \dots \cup V(P_i)]$.

\vspace{2ex}

\noindent
{\bf Claim 1:} {\em For $1 \le j \le q$, $S_j$ is a clique cutset in $T$ that separates $H_j \setminus S_j$ from $P_j \setminus S_j$.}
\\
\\
{\em Proof of Claim 1:}
To prove the claim we prove by induction on $i$ the following statement:
{\em for $1\le j \le i \le q$, $S_j$ is a cutset in $T_i$ that separates $H_j \setminus S_j$ from $P_j \setminus S_j$.}

By construction, the statement clearly holds for $i=1$.
Let $i > 1$, and inductively assume that for $1\le j\le i-1$, $S_j$ is a cutset in $T_{i-1}$ that separates $H_j \setminus S_j$ from $P_j \setminus S_j$.
Suppose that for some $j\le i$, $S_j$ is not a cutset in $T_i$ that separates $H_j \setminus S_j$ from $P_j \setminus S_j$.
Then clearly (by construction), $j\le i-1$.
By the induction hypothesis, $S_j$ is a cutset in $T_{i-1}$ that separates $H_j \setminus S_j$ from $P_j \setminus S_j$.
Let $C_{H_j}$ (respectively, $C_{P_j}$) be the connected component of $T_{i-1} \setminus S_j$ that contains $H_j \setminus S_j$ (respectively, $P_j \setminus S_j$).
Then, without loss of generality, $x_i \in C_{H_j}$ and $z_i \in C_{P_j}$. 
Since $H_i$ is a hole of $G_{i-1}$ that contains $x_i$ and $z_i$, it follows that $H_i \cap S_j = \{x_j, z_j\}$.
So, without loss of generality,  $y_i = x_j$. 
But then since $H_i \cup \{y_j\}$ cannot induce a theta in $G_{i-1}$, $(H_i, y_j)$ is a wheel $G_{i-1}$, 
which contradicts the fact that $P_i$ is a good ear. This completes the proof of Claim 1.

\vspace{2ex}

\noindent
{\bf Claim 2:} {\em $T$ is chordal and $\omega (T)\le 6$.}
\\
\\
{\em Proof of Claim 2:}
By Claim 1, it suffices to show that $T[H]$, and for $1\le i\le q$, $T[V(P_i)\cup \{ y_i\}]$ are all chordal and have maximum clique  sizes at most 6.
Let $G_0=G[H]$, and observe that, since $G$ is triangle-free,
for every $1\le i\le q$, every interior vertex of $P_i$ has at most one neighbor in $G_{i-1}$.

Let $H=v_1, \ldots ,v_k,v_1$, and  without loss of generality we assume that $u=v_1$ and $v=v_2$.
Suppose $C$ is a hole contained in $T[H]$. Since, by construction, $\{ u,v\}$ is complete to $V(H)\setminus \{ u,v\}$, $V(C)\cap \{ u,v\} =\emptyset$.
Let $v_i$ be the smallest-indexed vertex of $C$. So $i\ge 3$. It follows that the edges of $C$ are either edges of $H$ or are of type 1 and hence,
by the above observation, are short chords of $H$.
It follows by the choice of $v_i$ that $v_{i+1}$ and $v_{i+2}$ are the neighbors of $v_i$ in $C$. But then $v_{i+1}v_{i+2}$ is a chord of $C$, a contradiction.
Therefore $T[H]$ is chordal.
Since the edges of $T[V(H)\setminus \{ u,v\}]$ are either edges of $H$ or short chords of $H$, $\omega (T[V(H)\setminus \{ u,v\}] )\le 3$,
and hence $\omega (T[V(H)]) \le 5$.

Now consider an ear $P_i$.
Suppose that $T[V(P_i)\cup \{ y_i\}]$ contains a hole $C$. By construction, $S_i$ is a clique of $T[V(P_i)\cup \{ y_i\}]$
that is complete to $V(P_i)\setminus \{ x_i,z_i\}$, and hence $V(C)\cap S_i=\emptyset$.
So $V(C)\subseteq V(P_i)\setminus \{ x_i,z_i\}$. Recall that $P_i$ is a chordless path in $G$, and that every edge of $T[V(P_i)\setminus \{ x_i,z_i\}]$
that is not an edge of $P_i$ is of type 1. Let $x_jz_j$ be such a type 1 edge, and suppose that $y_j$ is not a vertex of $P_i$.
Then $j>i$. By the above observation, $y_j$ cannot be an interior vertex of $P_k$ where $k>i$.
So $y_j$ is a vertex of $G_{i-1}$, and since the only vertex of $G_{i-1}$ that can be adjacent in $G_i$ to two interior vertices of $P_i$ is $y_i$, it follows that
$y_j=y_i$. Let $H'$ be the hole obtained by
augmenting $H_i$ with $P_i$. Then the wheel $(H',y_j)$
is contained in $G_{j-1}$ and
contradicts $P_j$ being a good ear. Therefore $y_j\in V(P_i)$, and so every
type 1 edge of $T[V(P_i)\setminus \{ x_i,z_i\}]$ is a short chord of $P_i$. This contradicts the assumption that $C$ is a hole of $T[V(P_i)\setminus \{ x_i,z_i\}]$.
Hence, $T[V(P_i)\cup \{ y_i\}]$ is chordal.
Since $T[V(P_i)\setminus \{ x_i,z_i\}]$ consists of edges of $P_i$ and short chords of $P_i$, it follows that $\omega (T[V(P_i)\setminus \{ x_i,z_i\}]) \le 3$,
and therefore $\omega ([V(P_i)\cup \{ y_i\}]) \le 6$.
This completes the proof of Claim 2.

\vspace{2ex}

By Claim 2, it follows that there is a triangulation of $G^*$ with clique size at most 6, and hence the treewidth of $G^*$ is at most 5.
\end{proof}

As mentioned above, Corneil and Rotics \cite{cr} proved that the clique-width of a graph $G$  is at most $3\times 2^{tw(G)-1}$,
and so the following is a direct corollary of Theorem \ref{tftw5}.

\begin{corollary} \label{cwtehf}
Triangle-free odd-signable graphs have clique-width at most 48.
\end{corollary}





\begin{theorem} \label{cwcehf}
If $G$ is (cap, 4-hole)-free odd-signable graph with no clique cutset, then $G$ has clique-width at most 48.
\end{theorem}

\begin{proof}
Let  $G$ be a (cap, 4-hole)-free odd-signable graph with no clique cutset. Let $U$ be the set of universal vertices.
We may assume that $G\setminus U$ contains a hole, since otherwise
$G\setminus U$ is a clique and so its clique-width is 2.
Let $F$ be a maximal induced subgraph of $G\setminus U$ that is triangle-free and has no clique cutset.
By Theorem \ref{t2}, $G\setminus U$ is obtained from $F$ by substituting cliques for vertices of $F$. 
Since $F$ is triangle-free odd-signable, it follows from Corollary \ref{cwtehf} that the clique-width of $F$ is at most 48.
Substituting a graph $G_2$ for a vertex of a graph $G_1$ gives a graph with clique-width at most the maximum of the clique-widths of $G_1$ and $G_2$ \cite{co}, \cite{g}.
A clique of size at least 2 has clique-width 2.
Thus it follows that $G\setminus U$  has clique-width at most 48. Adding a universal vertex to a graph with at least one edge does not change the clique-width.
Thus $G$ has clique-width at most 48.
\end{proof}

\begin{theorem}\label{thm:c4htww}
If $G$ is a (cap, 4-hole)-free odd-signable graph with no clique cutset,
then $G$ has treewidth at most $6 \omega(G)-1$.
\end{theorem}

\begin{proof}
Let  $G$ be a (cap, 4-hole)-free odd-signable graph with no clique cutset. We may assume that $G$ contains a hole, since otherwise
$G$ is a clique and so the treewidth of $G$ is $|V(G)|-1=\omega(G)-1$. Let $U$ be the set of universal vertices of $G$, and note that
$G\setminus U$ has no clique cutset.
Let $F$ be a maximal induced subgraph of $G\setminus U$ that is triangle-free and has no clique cutset.
By Theorem \ref{t2}, $G\setminus U$ is obtained from $F$ by, for each vertex $v$, substituting a clique $K_v$. 
Since $F$ is triangle-free odd-signable,  it follows from Theorem \ref{tftw5} that the treewidth of $F$ is at most 5. 
In particular, there is a triangulation $T$ of $F$ with maximum clique size at most 6.
We can obtain a triangulation $T^\prime$ of $G\setminus U$ by substituting the cliques $K_v$ for the vertices $v$ of $T$.
Each of these cliques $K_v$ has size at most $\omega(G)-|U|$, so the size of a largest clique in $T^\prime$ is at most $6(\omega(G)-|U|)$.
We obtain a triangulation $T^{\prime\prime}$ of $G$ by adding to $T^\prime$ the clique $U$ 
and joining every vertex of $U$  to every vertex of $T^\prime$.
The largest clique in $T^{\prime\prime}$ has size at most $6 (\omega(G)-|U|)+|U| = 6 \omega(G)-5|U| \le 6 \omega(G)$. Thus $G$ has treewidth at most
$6 \omega(G)-1$.
\end{proof}

\section{Algorithms for coloring and maximum weight stable set} \label{algorithms}

In this section, we give polynomial-time algorithms for maximum weight stable set,
$q$-coloring (that is, coloring with a fixed number $q$ of colors), and chromatic number for (cap, 4-hole)-free odd-signable graphs (and in particular, (cap, even hole)-free graphs). Our algorithms will take the following general approach.

\begin{framed}
1. Decompose the input graph $G$ via clique cutsets into subgraphs that do not contain clique cutsets.
These subgraphs are called \emph{atoms}. \\

2. Find the solution for each atom using Theorems \ref{t2},  \ref{tftw5}, \ref{cwcehf}, \ref{thm:c4htww}.\\

3. Combine solutions to atoms along the clique cutsets to obtain a solution for $G$.
\end{framed}

\subsection{Clique cutset decomposition}

Let $G=(V,E)$ be a graph and $K\subseteq V$ a clique cutset such that
$G\setminus K$ is a disjoint union of two subgraphs $H_1$ and $H_2$ of $G$.
We let $G_i=H_i\cup K$ for $i=1,2$.
We say that $G$ is \emph{decomposed into} $G_1$ and $G_2$ \emph{via} $K$,
and call this a \emph{decomposition step}.
We then recursively decompose $G_1$ and $G_2$ via clique cutsets until no clique cutset exists.
This procedure can be represented by a rooted binary tree $T(G)$ where $G$ is the root and
the leaves are induced subgraphs of $G$ that do not contain clique cutsets. These subgraphs are called
\emph{atoms} of $G$.  Tarjan \cite{Ta85} showed that for any graph $G$, $T(G)$ can be found in
$O(nm)$ time. Moreover, in each decomposition step Tarjan's algorithm produces an atom,
and consequently $T(G)$ has at most $n-1$ leaves (or equivalently atoms).

Let $k\ge 1$ be a fixed integer.
Tarjan \cite{Ta85} observed that  $G$ is $k$-colorable if and only if each atom of $G$ is $k$-colorable.
This implies that if one can solve $q$-coloring or chromatic number for atoms, then one can also solve these problems
for $G$. It is straightforward to check that once a $k$-coloring of each atom is found, then it takes
$O(n^2)$ time to combine these colorings to obtain a $k$-coloring of $G$.

In a slightly more complicated fashion,  Tarjan \cite{Ta85} showed that once the maximum weight
stable set problem is solved for atoms, one can solve the problem for $G$.
Let $G=(V,E)$ be a graph with a weight function $w:V\rightarrow \mathbb{R}$. For a given subset $S\subseteq  V$,
we let $w(S)=\sum_{v\in S}w(v)$, and denote the maximum weight of a stable set of $G$ by $\alpha_w(G)$.
Suppose that $G$ is decomposed into $A$ and $B$  via a clique cutset $S$, where $A$ is an atom.
We explain Tarjan's approach as follows. To compute a stable set of weight $\alpha_w(G)$, we do the following.

(i) Compute a maximum weight  stable set $I'$ of $A\setminus S$.

(ii) For each vertex $v\in S$, compute a maximum weight  stable set $I_v$ of $A\setminus N[v]$.

(iii) Re-define the weight of $v\in S$ as $w'(v)=w(v)+w(I_v)-w(I')$.

(iv) Compute the maximum weight  stable set $I''$ of $B$ with respect to the new weight $w'$.
If $I''\cap S=\{v\}$, then let $I= I_v\cup I''$; otherwise let $I=I'\cup I''$.

It is easy to see that $\alpha_w(G)=w(I)$. This divide-and-conquer approach can be applied top-down
on $T(G)$ to obtain a solution for $G$ by solving $O(n^2)$ subproblems on induced subgraphs of atoms, as there are
$O(n)$ decomposition steps and each step amounts to solving $O(n)$ subproblems as explained in
(i)-(iv).

Therefore, it suffices to explain below how to solve coloring and maximum weight stable set for atoms of (cap, 4-hole)-free odd-signable graphs.

\subsection{Skeleton}

Let $G$ be a (cap, 4-hole)-free odd-signable graph without clique cutsets.
By Theorem \ref{t2}, $G$ is obtained from a (triangle, 4-hole)-free induced subgraph $F$ that has no clique cutset
by first blowing up vertices $v\in V(F)$ into  cliques $K_v$, and then adding a (possibly empty) universal clique $U$.
We call $F$ the \emph{skeleton} of $G$.

We say that two vertices $u$ and $v$ of $G$ are \emph{true twins} if $N_G[u]=N_G[v]$.
In particular, any pair of true twins are adjacent. It is clear that the binary relation on $V(G)$ defined
by being true twins is an equivalence relation and therefore $V(G)$ can be partitioned into equivalence classes of true twins.
Indeed, the cliques $K_v$ ($v\in V(F)$) and $U$ are equivalence classes of true twins.
Our algorithm relies on finding equivalence classes efficiently.
The following theorem is left as an exercise in \cite{Sp03}. We give a proof.
We say that a vertex $u$
\emph{distinguishes} vertices $v$ and $w$ if $u$ is adjacent to exactly one of $v$ and $w$.

\begin{theorem}\label{thm:twin linear time}
Given a graph $G$ with $n$ vertices and $m$ edges, one can find all equivalence classes of true twins in $O(n+m)$ time.
\end{theorem}
\begin{proof}
Suppose that $V(G)=\{v_1,v_2,\ldots,v_n\}$. We think of each vertex being adjacent to itself, and consequently
any vertex $v$ does not distinguish $v$ and any neighbor of $v$.
The idea is to start with the trivial partition $\mathcal{P}_0=\{V(G)\}$
and obtain a sequence of partitions $\mathcal{P}_1,\ldots, \mathcal{P}_n$ of $V(G)$
such that $\mathcal{P}_i$ is a refinement of $\mathcal{P}_{i-1}$ and is obtained as follows:
 for each set $S\in \mathcal{P}_{i-1}$, we partition $S$ into two subsets
$S'=S\cap N[v_i]$ and $S''=S\setminus S'$, and $\mathcal{P}_{i}=\cup_{S\in \mathcal{P}_{i-1}}\{S',S''\}$.
It can be easily proved by induction that for each $i$ it holds that
(i) any pair of vertices in a set $S\in \mathcal{P}_{i}$ are not distinguished by any of $v_1,\ldots,v_i$;
(ii) vertices from different sets in $\mathcal{P}_{i}$ are distinguished by one of $v_1,\ldots,v_i$.
Therefore, the equivalence classes of true twins are exactly the sets in $\mathcal{P}_n$.

It remains to show that this can be implemented in $O(n+m)$ time.
In the algorithm, we do not actually maintain the sets in a partition.
Instead, we use an array $s[v_i]$ to keep track of which subset $v_i$ belongs to.
Initially, we set $s[v_i]=0$ for all $i$ and this takes $O(n)$ time.
Then we do the following:
for each $1\le i\le n$, we  set $s[u]=s[u]+2^{i-1}$ for each $u\in N[v_i]$.
Clearly, this takes $\sum_{v_i}O(d(v_i))=O(m)$ time.
In the end, we group vertices with the same $s$-value by scanning the array once and this takes $O(n)$ time.
Therefore, the total running time is $O(n+m)$.
\end{proof}

In the following algorithms, we assume that $G$ is the input graph.
We first use Tarjan's algorithm to find $T(G)$ in $O(nm)$ time.
For any atom $A$ of $G$, we let $n_A$ and $m_A$ be the number of vertices
and the number of edges of $A$, respectively.
By Theorem \ref{thm:twin linear time} we can find the skeleton $F$ of $A$, $K_v$ and $U$
in $O(n_A+m_A)$ time.   Therefore, it takes $O(n)O(n+m)=O(nm)$ time to find
skeletons for all atoms of $G$.
So, we fix an atom $A$ and assume that the skeleton $F$ of $A$, $K_v$ ($v\in V(F)$) and $U$ are given.

\subsection{Solving chromatic number using clique-width}

It follows from Theorem \ref{tftw5} and Theorem \ref{cwcehf} that
$F$ has treewidth at most $5$ and clique-width at most $48$.
We first find a tree decomposition of $F$ with width at most $5$ in linear time  by Bodlaendar \cite{bod96},
and then feed this decomposition into the algorithm of Espelage, Gurski, and Wanke \cite{egw}
which outputs in linear time a $k$-expression of $F$ for some constant $k$ ($k$ could be larger than $48$).
Then we construct from $F$, $K_v$ ($v\in V(F)$) and $U$ in linear time a $k$-expression of  $G$ \cite{co}.
Finally, we find the chromatic number of $A$ in polynomial time by Kobler and Rotics \cite{kr}.
We solve chromatic number for every atom of $G$ in this way. 
The total running time is dominated by Kobler and Rotics's algorithm
\cite{kr} which runs in $O(2^{3k+1}k^2n^{2^{2k+1}+1})$ time.

\subsection{Solving $q$-coloring using treewidth}

We first find the clique number $\omega(G)$ of $G$ in $O(nm)$ time \cite{actv}.
If $\omega(G)>q$, then $G$ is not $q$-colorable, and we are done.
Otherwise, $\omega(G)\le q$ and so every atom $A$ also has $\omega(A)\le q$.
By  Theorem \ref{thm:c4htww},
the treewidth of $A$ is at most $6q-1$. We then use  Bodelander's algorithm \cite{bod96}
to find a tree decomposition with width $6q-1$ in $O(n_A)$ time. Finally,
$q$-coloring can be solved in $O(n_A)$ time for $A$ \cite{bk, courcelle}.
Since there are $O(n)$ atoms, the running time for find all colorings of atoms
is $O(n)O(n)=O(n^2)$. Recall that combining colorings of atoms can also
be done in $O(n^2)$ time, and so the total running time is dominated by
finding $T(G)$ and skeletons, that is, $O(nm)$.

\subsection{Solving maximum weight stable set using treewidth}

For maximum weight stable set, we let $v'\in K_v$ be the vertex with maximum weight
among vertices in $K_v$. Similarly, if $U\neq \emptyset$ then let $u'\in U$ be the vertex with largest weight among vertices
in $U$. Let $F'=\{u'\}\cup \{v':v\in V(F)\}$ if $U\neq \emptyset$,
and $F'=\{v':v\in V(F)\}$ if $U=\emptyset$.
Note that $F'$ is obtained from $F$ by adding at most one universal vertex. Moreover, the maximum weight
of a stable set in $A$ equals  the maximum weight of a stable set in $F'$.
It follows from Theorem \ref{tftw5} that $F$ has treewidth at most $5$,
and so $F'$ has treewidth at most $6$.
Let $S_A$ be the clique cutset used in the decomposition step that yields $A$, and
$n'_A=|V(A\setminus S_A)|$.
Recall that, for each atom $A$,
we need to solve $O(n)$ subproblems
on induced subgraphs of $A\setminus S_A$.
Each such subproblem can be solved in $O(n'_A)$ time by first finding
 a tree decomposition with width at most $6$ in $O(n'_A)$ time by Bodelander \cite{bod96},
and then solving maximum weight stable set  in $O(n'_A)$ time \cite{bk}.
Note that for two different atoms $A$ and $B$, the subgraphs of $A$ for which the subproblems
need to be solved
are vertex-disjoint
from the subgraphs of $B$ for which the subproblems need to be solved.
This implies that it takes $O(n)\sum_{A}n_A=O(n^2)$ time to
solve all these subproblems, where the summation goes over all atoms of $G$.
So, the total running time is dominated by finding $T(G)$ and skeletons, that is, $O(nm)$.
\\

An important feature of our algorithms is that they are \emph{robust} in the sense that we do
not need to assume that the input graph is (cap, 4-hole)-free odd-signable. Our algorithms
either report that the graph is not (cap, 4-hole)-free odd-signable or solve the problems
(in which case the input graph may or may not be (cap, 4-hole)-free odd-signable): for any input graph $G$,
we find the skeleton $F$ of any atom $A$ of $G$ and test if $F$ has treewidth at most $5$. If the answer
is no, then $G$ is not (cap, 4-hole)-free odd signable by Theorem \ref{tftw5};
otherwise we use the above algorithms to solve coloring or maximum weight stable set.

\subsection{Recognition}

Even-hole-free graphs were first shown to be recognizable in polynomial time in \cite{cckv-ehfrecognition}.
Currently, the fastest known recognition algorithm for this class has complexity $O(n^{11})$ \cite{cl}.
In \cite{cckv-tf}, an $O(n^4)$ algorithm is given for recognizing triangle-free odd-signable graphs
(and in particular (triangle, even hole)-free graphs).
In \cite{cckv-cap} an $O(n^6)$ algorithm is given for recognizing cap-free odd-signable graphs
(and in particular (cap, even hole)-free graphs).
We now show how to do this in $O(n^5)$-time.


\begin{lemma}\label{lem:detect cap}
There is an $O(nm^2)$ time algorithm to decide if a graph contains a cap.
\end{lemma}
\begin{proof}
We first guess an edge $e=uv$ and a vertex $w$ such that $w$ is the vertex that is adjacent to $u$ and $v$
which are in a hole not containing $w$. Clearly, there are $m$ choices for $e$ and at most $n$ choices for $w$.
We then test if $u$ and $v$ are in the same component of
$G'=G\setminus ((N[w]\setminus \{ u,v\} ) \cup (N(u)\cap N(v))\cup \{ e\} )$.
This can be done in $O(n+m)$ time using breadth-first search.
Therefore, the total running time is $O(m)O(n)O(n+m)=O(nm^2)$.
The correctness follows from the fact that
if there is a cap that consists of a hole $H$ going through $u$ and $v$, and the vertex $w$ that is not on $H$,
then there must exist a $uv$-path in $G'$.
\end{proof}


\begin{lemma}\label{lem:odd-signable=ehf}
Let $G$ be a graph that contains a universal vertex $u$. Then $G$ is odd-signable if and only if $G\setminus \{ u\}$ is even-hole-free.
\end{lemma}
\begin{proof}
Suppose that $G$ is odd-signable. Then $G$ does not contain thetas, prisms or even wheels by Theorem \ref{thm:forbid odd-signable}.
The fact that $u$ is universal implies that $G\setminus \{ u\}$ is even-hole-free.
Conversely, if $G\setminus \{u\}$ is even-hole-free, then clearly $G$ has no even wheels.
Furthermore, $G\setminus \{u\}$ has no prisms and thetas, as it is even-hole-free.
Note that  thetas and prisms do not contain universal vertices,  and so adding a universal vertex to $G\setminus \{u\}$
does not create a theta or prism. This shows that $G$ is odd-signable by Theorem \ref{thm:forbid odd-signable}.
\end{proof}


\begin{theorem}
There exists an $O(n^5)$ time algorithm to decide if a graph is (cap, 4-hole)-free odd-signable
(resp. (cap, even hole)-free).
\end{theorem}
\begin{proof}
Let $G$ be a graph. We first test if $G$ contains a 4-hole using brute force, and this takes $O(n^4)$ time.
If $G$ contains a 4-hole, then we stop. Therefore, we now assume that $G$ is 4-hole-free.
Secondly, we apply Lemma \ref{lem:detect cap} to see if $G$ contains a cap in $O(nm^2)=O(n^5)$ time.
If $G$ contains a cap, we stop. So, we may assume that $G$ is (cap, 4-hole)-free.

We then apply Tarjan's algorithm to find the clique cutset decomposition tree $T(G)$ in $O(nm)$ time.
Note that $G$ is odd-signable if and only if each atom is. For each atom $A$, we find its skeleton
$F$, $K_v$ for $v\in V(F)$ and $U$ in $O(n+m)$ time by Theorem \ref{thm:twin linear time}.
If $U=\emptyset$, then $A$ is odd-signable  if and only if $F$ is odd-signable, since adding twin vertices
preserves being odd-signable; if $U\neq \emptyset$, i.e., $A$ contains a universal vertex,
then it follows from Lemma \ref{lem:odd-signable=ehf} that $A$ is odd-signable if and only if
$F$ is even-hole-free. We finally apply the $O(n^4)$ time recognition algorithm from \cite{cckv-tf}
for triangle-free odd-signable graphs or (triangle, even hole)-free graphs to $F$ depending on whether
$U$ is empty or not. If the algorithm returns no for the skeleton $F$ of some atom,
then $G$ is not odd-signable; otherwise $G$ is odd-signable. The running time for testing all atoms
is $O(n)O(n^4)=O(n^5)$. Therefore, the total running time for recognizing (cap, 4-hole)-free odd-signable graphs is $O(n^5)$.
Similarly, (cap, even hole)-free graphs can be recognized with the same time complexity.
\end{proof}

\section{Open Problems} \label{problems}

The bound given by Corollary \ref{cor:binding} is attained by odd holes
and the Haj\'os graph (see Figure \ref{hajos}).
Note that these graphs have clique number at most $3$. For graphs
with large clique number, we do not have an example
showing that the bound is tight. Nevertheless, the optimal constant
is at least $\frac{5}{4}$. For any integer $k\ge 1$, let $G_k$ be the graph obtained from a $5$-hole by
substituting a clique of size $2k$ for each vertex of the $5$-hole.
Clearly, $|V(G_k)|=10k$, $\alpha(G_k)=2$ and $\omega(G_k)=4k$.
Hence, $\chi(G_k)\ge \frac{|V(G_k)|}{\alpha(G_k)}=5k$.
Moreover, it is easy to see that $G_k$ does admit a $5k$-coloring.
So,  $\chi(G_k)=5k=\frac{5}{4}\omega(G_k)$.
A natural question is  whether or not  one can reduce the constant from $\frac{3}{2}$ to $\frac{5}{4}$.


\begin{figure}
\center
\tikzstyle{every node}=[circle, draw, inner sep=0pt, minimum width=4pt]
\begin{tikzpicture}[scale=0.4]
\node [draw, circle] (v1) at (0,4) {};
\node [draw, circle] (v2) at (4,1) {};
\node [draw, circle] (v3) at (2,-3) {};
\node [draw, circle] (v4) at (-2,-3) {};
\node [draw, circle] (v5) at (-4,1) {};
\draw (v1) edge (v2);
\draw (v2) edge (v3);
\draw (v3) edge (v4);
\draw (v4) edge (v5);
\draw (v5) edge (v1);
\node [draw, circle] (v6) at (7,1) {};
\node [draw, circle] (v7) at (-7,1) {};
\draw (v1) edge (v6);
\draw (v6) edge (v3);
\draw (v4) edge (v7);
\draw (v7) edge (v1);
\draw (v2) edge (v6);
\draw (v5) edge (v7);
\end{tikzpicture}
\caption{\label{hajos} The Haj\'os graph.}
\end{figure}
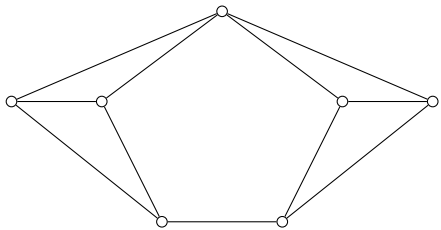


\vspace{0.5cm}

\noindent
{\bf Problem:} {\em Is it true that $\chi(G)\le \lceil\frac{5}{4}\omega(G)\rceil$ for every (cap, even hole)-free graph $G$?}
\\

It was shown in \cite{cks} that this is true for the class of $(C_4,P_5)$-free graphs, 
which is a subclass of the class of (cap, even hole)-free graphs.

Even-hole-free graphs are also known to be $\chi$-bounded.
In \cite{achrs} it is shown that every even-hole-free graph has a vertex whose neighborhood is a union of two
cliques. This implies that if $G$ is even-hole-free, then $\chi (G)\leq 2\omega (G)-1$.
It remains open whether a better bound is possible.

The complexity of 3-coloring, $q$-coloring, and minimum coloring is unknown
for even-hole-free graphs, 4-hole-free odd-signable graphs, and odd-signable graphs.
Polynomial-time algorithms for minimum coloring have been given for (diamond, even hole)-free graphs \cite{kmv}
and (pan, even hole)-free graphs \cite{cch}.

The clique covering problem is to find a minimum number of cliques which partition the vertices of a graph.
This problem is the same as finding a minimum coloring of the complementary graph.
The complexity of this problem is unknown for the following classes of graphs:
(cap, even hole)-free graphs, (cap, 4-hole)-free odd-signable graphs, 4-hole-free odd-signable graphs, even-hole-free graphs,
and odd-signable graphs.\\

\textbf{Acknowledgement}. We would like to thank Haiko M{\"u}ller for fruitful discussions,
and to thank Jerry Spinrad for pointing to us the idea of finding true twins in $O(n+m)$ time.



\begin{thebibliography}{99}

\bibitem{actv}
P. Aboulker, P. Charbit, N. Trotignon and K. Vu\v{s}kovi\'c,
\emph{Vertex elimination orderings for hereditary graph classes}
Discrete Mathematics 338 (2015) no. 5, 825-834.

\bibitem{achrs}
L. Addario-Berry, M. Chudnovsky, F. Havet, B. Reed and P. Seymour,
\emph{Bisimplicial vertices in even-hole-free graphs},
Journal of Combinatorial Theory B 98 (2008) 1119-1164.

\bibitem{bod96}
H. L. Bodlaender,
\emph{A linear-time algorithm for finding tree-decompositions of small treewidth},
SIAM J. Comput. 25 (1996) no. 6, 1305-1317.


\bibitem{bk}
H. L. Bodlaender and A. M. C. A. Koster,
\emph{Combinatorial optimization on graphs of bounded treewidth},
The Computer Journal 51 (2008) no. 3, 255-269.


\bibitem{bf}
M. Burlet and J. Fonlupt,
\emph{Polynomial algorithm to recognize a Meyniel graph},
Discrete Mathematics 21 (1984) 225-252.


\bibitem{ce}
K. Cameron and J. Edmonds,
\emph{An algorithm for finding a strong stable set or a Meyniel obstruction in any graph},
Discrete Mathematics and Theoretical Computer Science Proceedings (2005) 203-206.	

\bibitem{clm}
K. Cameron, B. L\'ev\^eque and F. Maffray,
\emph{Coloring vertices of a graph or finding a Meyniel obstruction},
Theoret. Comput. Sci. 428 (2012) 10-17.

\bibitem{cch}
K. Cameron, S. Chaplick and C. T. Ho\`ang,
\emph{On the structure of {pan, even hole}-free graphs},
arXiv:1508.03062 [cs.DM], 2015,
to appear in Journal of Graph Theory.

\bibitem{cl}
H.-C. Chang and H.-I. Lu,
\emph{A faster algorithm for recognising even-hole-free graphs},
Journal of Combinatorial Theory B 113 (2015) 141-161.

\bibitem{cks}
S.A.~Choudum, T.~Karthick, M.A.~Shalu,
\emph{Perfectly coloring and linearly $\chi$-bound ${P}_6$-free graphs},
Journal of Graph Theory (2007) 293-306.

\bibitem{cckv-cap}
M. Conforti, G. Cornu\'ejols, A. Kapoor and K. Vu\v{s}kovi\'c,
\emph{Even and odd holes in cap-free graphs},
Journal of Graph Theory 30 (1999) 289-308.

\bibitem{cckv-tf}
M. Conforti, G. Cornu\'ejols, A. Kapoor and K. Vu\v{s}kovi\'c,
\emph{Triangle-free graphs that are signable without even holes},
Journal of Graph Theory 34 (2000) 204-220.

\bibitem{cckv-ehfrecognition}
M. Conforti, G. Cornu\'ejols, A. Kapoor and K. Vu\v{s}kovi\'c,
\emph{Even-hole-free graphs, Part II: Recognition algorithm},
Journal of Graph Theory 40 (2002) 238-266.

\bibitem{cgp}
M. Conforti, B. Gerards and K. Pashkovich,
\emph{Stable sets and graphs with no even hole},
Math. Program., Ser. B 153 (2015) 13-39.

\bibitem{cr}
D.G. Corneil and U. Rotics,
\emph{On the relationship between clique-width and treewidth}, SIAM J. Comput. 34 (2005) no. 4, 825-847.

\bibitem{cc}
G. Cornu\'ejols and W.H. Cunningham,
\emph{Compositions for perfect graphs}, Discrete Mathematics 55 (1985) 245-254.

\bibitem{courcelle}
B. Courcelle,
\emph{The monadic second-order logic of graphs. i. Recognizable sets of finite graphs},
Information and computation 85 (1990) no. 1, 12-75.

\bibitem{cmr}
B. Courcelle, J. Makowsky, and U. Rotics,
\emph{Linear time solvable optimization problems on graphs of bounded clique-width},
Theory of Computing Systems 33 (2000) no. 2, 125-150.

\bibitem{co}
B. Courcelle and S. Olariu,
\emph{Upper bounds to the clique width of graphs},
Discrete Appl. Math. 101 (2000) no. 1-3, 77-114.

\bibitem{daSV}
M.V.G. da Silva and K. Vu\v{s}kovi\'c,
\emph{Triangulated neighborhoods in even-hiole-free graphs},
Discrete Mathematics 307 (2007) 1065-1073.


\bibitem{egw}
W. Espelage, F. Gurski, and E. Wanke,
\emph{Deciding clique-width for graphs of bounded tree-width},
J. Graph Algorithms Appl. 7 (2003) no. 2, 141-180.

\bibitem{g}
F. Gurski,
\emph{Graph operations on clique-width bounded graphs},
arXiv:cs/0701185v3 [cs.DS], 2016.

\bibitem{hertz}
A. Hertz,
\emph{A fast algorithm for colouring Meyniel graphs},
Journal of Combinatorial Theory B 50 (1990) 231-240.

\bibitem{hoangM}
C. T. Ho\`ang,
\emph{On a conjecture of Meyniel},
J.~Comb.~Th.~B 42 (1987) 302-312.



\bibitem{kmv}
T. Kloks, H. M{\"u}ller and K. Vu\v{s}kovi\'c,
\emph{Even-hole-free graphs that do not contain diamonds: a structure theorem and its consequences},
J. Combin. Theory Ser. B 99 (2009) no. 5, 733�800.

\bibitem{kr}
D. Kobler and U. Rotics,
\emph{Edge dominating set and colorings on graphs with fixed clique-width}, Discrete Appl. Math. 126 (2003) no. 2-3, 197-221.

\bibitem{mk}
S.E.~Markosjan and I.A.~Karapetjan,
\emph{Perfect graphs},
Akad. Nauk Armjan. SSR. Dokl. 63 (1976) 292-296.

\bibitem{mgr}
S.E. Markossian, G.S. Gasparian and B.A. Reed,
\emph{$\beta$-perfect graphs},
Journal of Combinatorial Theory B 67 (1996) 1-11.

\bibitem{meyniel}
H. Meyniel,
\emph{On the perfect graph conjecture},
Discrete Mathematics 16 (1976) 339-342.


\bibitem{oum}
S. Oum,
\emph{Approximating rank-width and clique-width quickly},
ACM Trans. Algorithms 5 (2009) no. 1, Art. 10, 20 pp.

\bibitem{os}
S. Oum and P. D. Seymour,
\emph{Approximating clique-width and branchwidth},
Journal of Combinatorial Theory, Series B, 96 (2006) no. 4, 514-528.

\bibitem{rr}
F. Roussel and I. Rusu,
\emph{Holes and dominoes in Meyniel graphs},
Internat. J. Found. Comput. Sci. 10 (1999) 127-146.

\bibitem{rr2}
F. Roussel and I. Rusu,
\emph{An $O(n^2)$ algorithm to color  Meyniel graphs},
Discrete Mathematics (2001) 107-123.

\bibitem{sss}
A. Silva, A.A. da Silva and C.L. Sales,
\emph{A bound on tree width of planar even-hole-free graphs},
Discrete Applied Mathematics 158 (12) 1229-1239.

\bibitem{Sp03} J.P. Spinrad, \emph{Efficient Graph Representations}, American Mathematical Society, 2003.

\bibitem{Ta85}
R.E. Tarjan,
\emph{Decomposition by clique separators},
Discrete Mathematics (1985) 221-232.

\bibitem{truemper}
K. Truemper,
\emph{Alpha-balanced graphs and matrices and GF(3)-representability of matroids},
Journal of Combinatorial Theory B 32 (1982) 112-139.

\bibitem{v}
K. Vu\v{s}kovi\'c,
\emph{Even-hole-free graphs: A survey},
Appl. Anal. Discrete Math. 4 (2010) 219-240.

\end{thebibliography}
\end{document}